\documentclass[12pt,a4paper]{article}
\usepackage{amsmath,amsthm,amssymb,amscd}
\usepackage{latexsym}
\usepackage{indentfirst}
\usepackage{bibentry}
\usepackage{textcomp}
\usepackage{float}
\usepackage{multirow}
\usepackage{booktabs}
\usepackage{graphicx}
\usepackage{color}
\usepackage{cite}
\usepackage{authblk}
\makeatletter \@addtoreset{equation}{section}

\makeatother
\newtheorem{thm}{Theorem}[section]
\newtheorem{cor}{Corollary}[section]
\newtheorem{lem}{Lemma}[section]
\newtheorem{exmp}{Example}[section]

\theoremstyle{definition}
\newtheorem{rem}{Remark}[section]

\begin{document}

\title{The Smallest Eigenvalue of Large Hankel Matrices Generated by a Deformed Laguerre Weight}

\author[1]{{Mengkun Zhu}\footnote{Zhu\_mengkun@163.com}}
\author[2]{{Niall Emmart}\footnote{nemmart@yrrid.com}}
\author[1]{{Yang Chen}\footnote{yangbrookchen@yahoo.co.uk}}
\author[2]{{Charles Weems}\footnote{weems@cs.umass.edu}}

\affil[1]{Department of Mathematics, University of Macau,
Avenida da Universidade, Taipa, Macau, China}

\affil[2]{College of Information and Computer Sciences, \protect\\
University of Massachusetts, Amherst, MA 01003, USA}

\renewcommand\Authands{ and }

\maketitle

\begin{abstract}

We study the asymptotic behavior of the smallest eigenvalue, $\lambda_{N}$, of the Hankel (or moments) matrix denoted by $\mathcal{H}_{N}=\left(\mu_{m+n}\right)_{0\leq m,n\leq N}$, with respect to the weight $w(x)=x^{\alpha}{\rm e}^{-x^{\beta}},~x\in[0,\infty),~\alpha>-1,~\beta>\frac{1}{2}$. Based on the research by Szeg\"{o}, Chen, etc., we obtain an asymptotic expression of the orthonormal polynomials $\mathcal{P}_{N}(z)$ as $N\rightarrow\infty$, associated with $w(x)$. Using this, we obtain the specific asymptotic formulas of  $\lambda_{N}$ in this paper.

Applying the parallel algorithm discovered by Emmart, Chen and Weems, we get a variety of numerical results of $\lambda_{N}$ corresponding to our theoretical calculations.
\end{abstract}

\section{Introduction}
Random matrix theory (RMT) originated in multivariate statistics in the work of Hsu, Wishart and others in the 1930s (see the monograph \cite{C26}). In 1950s, Wigner put forward similar models for the regularity observed in the energy level distribution of heavy nuclei, where the energy levels are the eigenvalues of large random matrices. From the 1960s to 1970s, through the fundamental work of Dyson, Mehta, Gaudin, des Cloizeaux, Widom, Tracy, Wilf and others, RMT developed into a branch of Mathematical Physics. Its rapid development from the 1990s is due a string of fundamental discoveries of Tracy and Widom on the probability laws governing the largest and smallest eigenvalues of two families of Hermitian random matrices, the Gaussian Unitary Ensembles (GUE) and the Laguerre Unitary Ensembles (LUE).

RMT plays an important role in many diverse fields, multivariate statistics, quantum physics, Multi-Input-Multi-Output (MIMO) wireless communication, and stock movements in financial markets, etc. For a variety of theories and applications of RMT, see \cite{new1,new7,new8,C17,new4,new5,C21,new9,new6} and related references therein. RMT considers the properties, e.g. determinants, eigenvalues, eigenvalue distributions, eigenvectors, spectra, inverse, etc., of matrices whose elements are random variables chosen from a given distribution.

The analysis of Hankel matrices, occurs naturally in moment problems, which plays an important role in RMT. On moment problems, please see the monographs by Akhiezer \cite{C16} and by Krein \cite{new11}. The study of the largest and smallest eigenvalues are important since they provide useful information about the nature of the Hankel matrix generated by a given weight function, e.g. they are related with the inversion of Hankel matrices, where the condition numbers are enormously large.

Given $\left\{\mu_{k}\right\}$ the moment sequence of a weight function $w(x)(>0)$ with infinite support $s$,
\begin{equation}\label{b2}
\mu_{k}:=\int_{s} x^{k}w(x)dx,~~k=0,1,2,\ldots,
\end{equation}
the Hankel matrices, it is known that
\begin{equation}\label{b1}
\mathcal{H}_{N}:=\left(\mu_{m+n}\right)_{m,n=0}^{N},~~ N=0,1,2,\ldots
\end{equation}
are positive definite, see \cite{zzz}.

Let $\lambda_{N}$ denote the smallest eigenvalue of $\mathcal{H}_{N}$. The asymptotic behavior of $\lambda_{N}$ for large $N$ has been investigated in \cite{C3,C5,C6,new33,C9,C10,C11,C13,C14,C22}. Also see \cite{C24,C25}, in which the authors have studied the behavior of the condition number $\kappa\left(\mathcal{H}_{N}\right):=\frac{\Lambda_{N}}{\lambda_{N}}$, where $\Lambda_{N}$ denotes the largest eigenvalue of $\mathcal{H}_{N}$.

Szeg\"{o} \cite{C3} studied the asymptotic behavior of $\lambda_{N}$ for the Hermite (or Gaussion) weight ($w(x)={\rm e}^{-x^{2}},x\in\mathbb{R}$) and the Laguerre weight ($w(x)={\rm e}^{-x},x\geq0$). He found\footnote[1]{ In all of this paper, $a_{N}\simeq b_{N}$ means $\lim_{N\rightarrow\infty}a_{N}/b_{N}$=1.}
\begin{equation*}
\lambda_{N}\simeq A N^{\frac{1}{4}}B^{\sqrt{N}},
\end{equation*}
where $A,B$ are certain constants, satisfying $0<A,~0<B<1$. Also, Szeg\"{o} \cite{C3} showed that the largest eigenvalue $\Lambda_{N}$ corresponding to the Hankel matrices $\left[\frac{1}{i+j+1}\right]_{i,j=0}^{N}$, $\left[\Gamma\left(\frac{i+j+1}{2}\right)\right]_{i,j=0}^{N}$ and $\left[\Gamma(i+j+1)\right]_{i,j=0}^{N}$ were approximated by  $\frac{\pi}{2}$, $\Gamma\left(N+\frac{1}{2}\right)$ and $(2N)!$ respectively.

In \cite{C6}, Widom and Wilf investigated the case where $w(x)$ is supported in a compact interval $[a,b]$, such that the Szeg\"{o} condition
\begin{equation}
\int_{a}^{b}\frac{\ln w(x)}{\sqrt{(b-x)(x-a)}}dx>-\infty,
\end{equation}
holds, then they obtained
\begin{equation*}
\lambda_{N}\simeq A \sqrt{N}B^{N}.
\end{equation*}

Chen and Lawrence \cite{C9} found the asymptotic behavior of $\lambda_{N}$ with the weight function $w(x)={\rm e}^{-x^{\beta}},~x\in[0,\infty),~\beta>\frac{1}{2}$. Berg, Chen and Ismail \cite{C13} proved that the moment sequence (\ref{b2}) is determinate iff $\lambda_{N}\rightarrow0$ as $N\rightarrow\infty$. This is a new criteria for the determinacy of the Hamburger moment problem. Also, in the same paper, they obtained a lower bound of $\lambda_{N}$ for large $N$. In \cite{C10}, Chen and Lubinsky obtained the behavior of $\lambda_{N}$ when $w(x)={\rm e}^{-|x|^{\alpha}},~x\in\mathbb{R},~\alpha>1$. Berg and Szwarc \cite{C14} proved that $\lambda_{N}$ has exponential decay to zero for any measure which with compact support.

Zhu, Chen, Emmart and Weems \cite{C22} studied the Jacobi case, i.e. $w(x)=x^{\alpha}(1-x)^{\beta},~x\in[0,1],~\alpha>-1,~\beta>-1$ and provided a asymptotic behavior of $\lambda_{N}$,
\begin{equation*}
\lambda_{N}\simeq2^{\frac{15}{4}}\pi^{\frac{3}{2}}\left(1+2^{\frac{1}{2}}\right)^{-2\alpha}\left(1+2^{-\frac{1}{2}}\right)^{-2\beta}N^{\frac{1}{2}}\left(1+2^{\frac{1}{2}}\right)^{-4(N+1)},
\end{equation*}
which reduces to Sezg\"{o}'s result \cite{C3}, if $\alpha=\beta=0$.

The examples above show that the values of $\lambda_{N},N\rightarrow\infty$ are exponentially small, and the asymptotic behavior of $\lambda_{N}$ depends on the $w(x)$ in a non-trivial way. We are motivated by this phenomenon and the purpose of this paper is again to study the asymptotic behavior of $\lambda_{N}$, here we choose an generalised Laguerre weight $w(x)=x^{\alpha}{\rm e}^{-x^{\beta}},~x\in[0,\infty),~\alpha>-1,~\beta>\frac{1}{2}$.

The remainder of this paper is organized in 5 sections. In section 2 we reproduce some known results (Refs. \cite{C3,C9,C10,C13}, etc.) that will be applied to find the estimation of $\lambda_{N}$. In section 3, by adopting a previous result \cite{C20}, we obtain the asymptotic formula for the polynomials orthonormal with respect to $w(x)=x^{\alpha}{\rm e}^{-x^{\beta}},~x\in[0,\infty),~\alpha>-1,~\beta>\frac{1}{2}$, which is then employed in sections 4 and 5 for the determination of the large $N$ behavior of $\lambda_{N}$.
And finally, in section 6, we present a comparison of the theoretical results to numeric calculations for the smallest eigenvalue, for various values of $\alpha$, $\beta$ and $N$. The numerical computations were performed using the parallel algorithms developed in \cite{C11}.

\section{Preliminaries}
Consider the weight
\begin{equation*}
w(x):=x^{\alpha}{\rm e}^{-x^{\beta}},~x\in[0,\infty),~\alpha>-1,~\beta>\frac{1}{2},
\end{equation*}
in this case, the moments are
\begin{equation*}
\mu_{n}:=\int_{0}^{\infty}x^{n}w(x)dx=\frac{1}{\beta}\Gamma\left(\frac{1+\alpha+n}{\beta}\right),
\end{equation*}
and the positive Hankel matrix is
\begin{equation*}
\mathcal{H}_{N}:=\left(\mu_{m+n}\right)_{m,n=0}^{N}.
\end{equation*}
The focus of this paper is to derive the asymptotic behavior of the smallest eigenvalue $\lambda_{N}$ of $\mathcal{H}_{N}$.

It is well known that the smallest eigenvalue $\lambda_{N}$ can be found using the classical Rayleigh quotient
\begin{equation}\label{e1}
\lambda_{N}=\min\left\{\frac{\sum_{m,n=0}^{N}\overline{x}_{m}\mu_{m+n}x_{n}}{\sum_{n=0}^{N}|x_{n}|^{2}}~\Bigg{|}~X:=\left(x_{0},x_{1},\ldots,x_{N}\right)^{T}\in\mathbb{C}^{N+1}
\setminus\{0\}\right\}.
\end{equation}
Let $P_{N}$ be the orthogonal polynomials associated with $w(x)$, and denote by
\begin{equation*}
P_{N}(z):=\sum_{n=0}^{N}x_{n}z^{n},
\end{equation*}
then
\begin{equation}\label{0.00}
\int_{0}^{\infty}|P_{N}(x)|^{2}w(x)dx=\sum_{m,n=0}^{N}\overline{x}_{m}\mu_{m+n}x_{n},
\end{equation}

If we denote the orthonomal polynomials associated with the weight $w(x)$ by $\mathcal{P}_{N}(x)$, through
\begin{equation*}
P_{N}(z)=\sqrt{h_{N}}\mathcal{P}_{N}(z),
\end{equation*}
where $h_{N}$ is the square of the $L^{2}$ norm of $P_{N}(z)$, such that
\begin{equation}\label{200}
\int_{0}^{\infty}|\mathcal{P}_{N}(x)|^{2}w(x)dx=1,
\end{equation}
then the expression for $\lambda_{N}$, (\ref{e1}), can be recast as
\begin{equation}\label{e2}
\lambda_{N}=\min\left\{\frac{2\pi}{\int_{-\pi}^{\pi} \left|\mathcal{P}_{N}({\rm e}^{{\rm i}\theta})\right|^{2}d\theta} \right\}.
\end{equation}

If we define
\begin{equation*}
P_{N}(z):=\sum_{n=0}^{N}\xi_{n}\mathcal{P}_{n}(z),~~~{\rm and}~~~\mathcal{K}_{mn}:=\int_{-\pi}^{\pi}\mathcal{P}_{m}\left({\rm e}^{{\rm i}\theta}\right)\mathcal{P}_{n}\left({\rm e}^{-{\rm i}\theta}\right)d\theta,
\end{equation*}
we can see that
\begin{equation*}
\int_{-\pi}^{\pi}\left|P_{N}\left({\rm e}^{{\rm i}\theta}\right)\right|^{2}d\theta=\sum_{m,n=0}^{N}\overline{\xi}_{m}\mathcal{K}_{mn}\xi_{n},
\end{equation*}
Hence, the formula (\ref{e2}) will be equivalent to
\begin{equation}\label{e3}
\lambda_{N}=\min\left\{\frac{2\pi}{\sum_{m,n=0}^{N}\overline{\xi}_{m}\mathcal{K}_{mn}\xi_{n}}:\sum_{n=0}^{N}\left|\xi_{n}\right|^{2}=1  \right\}.
\end{equation}
Based on the Cauchy-Schwarz inequality, we will find that

\begin{equation*}
\begin{split}
\sum_{m,n=0}^{N}\overline{\xi}_{m}\mathcal{K}_{mn}\xi_{n}\leq\sum_{m,n=0}^{N}\mathcal{K}_{mm}^{\frac{1}{2}}\mathcal{K}_{nn}^{\frac{1}{2}}\left|\xi_{m}\right|\left|
\xi_{n}\right|\leq\sum_{m=0}^{N}\mathcal{K}_{mm}\cdot\sum_{n=0}^{N}\left|\xi_{n}\right|^{2}
\leq\sum_{n=0}^{N}\mathcal{K}_{nn}.
\end{split}
\end{equation*}
Therefore, a lower bound for the smallest eigenvalue of $\lambda_{N}$ is given by
\begin{equation}\label{e4}
\lambda_{N}\geq\frac{2\pi}{\sum_{n=0}^{N}\mathcal{K}_{nn}}.
\end{equation}

\section{The orthonomal polynomials with respect to the weight $w(x)=x^{\alpha}{\rm e}^{-x^{\beta}}$.}

The purpose of this section is to find the asymptotics of the orthonomal polynomials $\{\mathcal{P}_{N}(z)\}$ with respect to the weight $w(x)=x^{\alpha}{\rm e}^{-x^{\beta}},~x\in[0,\infty),~\alpha>-1,~\beta>\frac{1}{2}$.
Based on the Coulomb fluid linear statistics method, it has been proved in \cite{C20}, for $N\rightarrow\infty$, that the monic orthogonal polynomials $P_{N}(z)$ associated with $w(x)={\rm e}^{-v(x)}$ can be approximated by
\begin{equation}\label{pn}
P_{N}(z)\simeq\exp\left[-S_{1}(z)-S_{2}(z)\right],
\end{equation}
where
\begin{equation*}
S_{1}(z)=\frac{1}{4}\ln\left[\frac{16(z-a)(z-b)}{(b-a)^{2}}\left(\frac{\sqrt{z-a}-\sqrt{z-b}}{\sqrt{z-a}+\sqrt{z-b}}\right)^{2}\right],~\text{$z\notin [a,b]$},
\end{equation*}

\begin{equation*}
\begin{split}
S_{2}(z)  &=-N\ln\left(\frac{\sqrt{z-a}+\sqrt{z-b}}{2}\right)^{2}\\
         &  +\frac{1}{2\pi}\int_{a}^{b}\frac{v(x)}{\sqrt{(b-x)(x-a)}}\left[\frac{\sqrt{(z-a)(z-b)}}{x-z}+1\right]dx,~\text{$z\notin [a,b]$}.
\end{split}
\end{equation*}
Chen and his co-authors\cite{Chenchen} also gave an equivalent representation for $S_{1}$:
\begin{equation*}
{\rm e}^{-S_{1}(z)}=\frac{1}{2}\left[\left(\frac{z-b}{z-a}\right)^{\frac{1}{4}}+\left(\frac{z-a}{z-b}\right)^{\frac{1}{4}}\right],~\text{$z\notin [a,b]$}.
\end{equation*}
Consequently, we have,
\begin{thm}\label{1111}
For $N\rightarrow\infty$, the orthonomal polynomials associated with the weight $w(x)=x^{\alpha}{\rm e}^{-x^{\beta}},~x\in[0,\infty),~\alpha>-1,~\beta>\frac{1}{2}$ are approximated by
\begin{equation*}
\mathcal{P}_{N}(z)\simeq(-z)^{-\frac{\alpha}{2}}\frac{(-1)^{N}}{\sqrt{2\pi b}}\frac{\exp\left[-\mathcal{I}(z)+(2N+1+\alpha)\log\left(\sqrt{\eta}+\sqrt{\eta+1}\right)\right]}{\left[\eta(1+\eta)\right]^{\frac{1}{4}}},
\end{equation*}
with
\begin{equation}\label{0.01}
\begin{split}
\mathcal{I}(z):&=-\frac{2N+\alpha}{2\beta-1}\sqrt{\eta(1+\eta)}\cdot{_{2}F_{1}}\left(1,1-\beta;\frac{3}{2}-\beta;-\eta\right)-\frac{(-z)^{\beta}}{2}\cdot\sec(\pi\beta)\\
&=-\frac{2N+\alpha}{2\beta}\sqrt{\frac{\eta}{1+\eta}}
\cdot{_{2}F_{1}}\left(1,\frac{1}{2};1+\beta;\frac{1}{1+\eta}\right),\\
\end{split}
\end{equation}
where $z\notin[0,b]$ and $\eta:=-\frac{z}{b}$, whilst
\begin{equation*}
b:=C(2N+\alpha)^{\frac{1}{\beta}} ~~~{ and} ~~~C=C(\beta):=4\left[\frac{\Gamma(\beta+1)\Gamma(\beta)}{\Gamma(2\beta+1)}\right]^{\frac{1}{\beta}}.
\end{equation*}
\end{thm}
\begin{proof}
For our problem, $a=0$, whilst $b(N,\alpha,\beta)$ follows from the supplementary condition \cite{C20,C23}
\begin{equation*}
\int_{a}^{b}\frac{xv'(x)}{\sqrt{(b-x)(x-a)}}=2\pi N,
\end{equation*}
where
 \begin{equation*}
v(x)=-\ln w(x)=-\alpha\ln x+x^{\beta}.
\end{equation*}
Hence we have
\begin{equation*}
b=C(2N+\alpha)^{\frac{1}{\beta}} ~~~{\rm with} ~~~C:=4\left[\frac{\Gamma(\beta+1)\Gamma(\beta)}{\Gamma(2\beta+1)}\right]^{\frac{1}{\beta}}.
\end{equation*}

Let $\eta:=-\frac{z}{b},~z\notin[0,b]$, by taking the branch $-b\eta=b\eta{\rm e}^{{\rm i}\pi}$, $-b\eta-b=b(1+\eta){\rm e}^{{\rm i}\pi}$ we have
\begin{equation*}
\begin{split}
-S_{1}(z) =\ln\left[2^{-1}\cdot\left(\eta(\eta+1)\right)^{-\frac{1}{4}}\left(\sqrt{\eta+1}+\sqrt{\eta}\right)\right],\\
\end{split}
\end{equation*}
\begin{equation*}
\begin{split}
-S_{2}(z)  &=N\ln\left(\frac{\sqrt{-b\eta}+\sqrt{-b\eta-b}}{2}\right)^{2}-\frac{1}{2\pi}\int_{0}^{b}\frac{-\alpha\ln x+x^{\beta}}{\sqrt{(b-x)x}}\left[\frac{\sqrt{z(z-b)}}{x-z}+1\right]dx\\
         &=N\ln\frac{-b\left(\sqrt{\eta}+\sqrt{\eta+1}\right)^{2}}{4}-f(z)-K,\\
\end{split}
\end{equation*}
where
\begin{equation*}
\begin{split}
K:=\frac{1}{2\pi}\int_{0}^{b}\frac{-\alpha\ln x+x^{\beta}}{\sqrt{(b-x)x}}dx=-\frac{\alpha}{2}\ln b+\alpha\ln 2+\frac{2N+\alpha}{2\beta},
\end{split}
\end{equation*}
and $f(z)$ is defined by
\begin{equation}\label{0.1}
f(z):=\frac{\sqrt{z(z-b)}}{2\pi}\int_{0}^{b}\frac{-\alpha\ln x+x^{\beta}}{\sqrt{x(b-x)}(x-z)}dx,~\text{$z\notin [0,b]$}.
\end{equation}
Next, we will focus on the explicit formula of $f(z)$. From (\ref{0.1}), we have
\begin{equation*}
\begin{split}
&f(z)  :=\mathcal{I}_{1}(z)+\mathcal{I}(z)\\
         &=:-\frac{\alpha\sqrt{z(z-b)}}{2\pi}\int_{0}^{b}\frac{\ln x}{\sqrt{x(b-x)}(x-z)}dx+\frac{\sqrt{z(z-b)}}{2\pi}\int_{0}^{b}\frac{x^{\beta}}{\sqrt{x(b-x)}(x-z)}dx.\\
\end{split}
\end{equation*}
With the aid of the integral identities in the Appendix, we get
\begin{equation*}
\begin{split}
\mathcal{I}_{1}(z)=-\frac{\alpha\sqrt{z(z-b)}}{2\pi}\int_{0}^{b}\frac{\ln x}{(x-z)\sqrt{x(b-x)}}dx=\frac{\alpha}{2}\ln (-z)-\alpha\ln\left(\sqrt{\eta+1}+\sqrt{\eta}\right).
\end{split}
\end{equation*}
From the definition and basic properties of the Hypergeometric function \cite{new2},
\begin{equation*}
\begin{split}
\mathcal{I}(z)&=\frac{\sqrt{z(z-b)}}{2\pi}\int_{0}^{1}\frac{(by)^{\beta}}{(by-z)\sqrt{by(b-by)}}bdy\\
&=-\sqrt{\frac{z(z-b)}{\pi}}\cdot\frac{b^{\beta}}{2z}\cdot \frac{\Gamma\left(\frac{1}{2}+\beta\right)}{\Gamma(1+\beta)}\cdot{_{2}F_{1}}\left(1,\frac{1}{2}+\beta;1+\beta;-\frac{1}{\eta}\right)\\
&=-\frac{2N+\alpha}{2\beta-1}\sqrt{\eta(1+\eta)}\cdot{_{2}F_{1}}\left(1,1-\beta;\frac{3}{2}-\beta;-\eta\right)-\frac{(-z)^{\beta}}{2}\cdot\sec(\pi\beta)\\
&=-\frac{2N+\alpha}{2\beta}\sqrt{\frac{\eta}{1+\eta}}
\cdot{_{2}F_{1}}\left(1,\frac{1}{2};1+\beta;\frac{1}{1+\eta}\right).\\
\end{split}
\end{equation*}
Consequently, by (\ref{pn}), the monic orthogonal polynomials can be obtained as follows:
\begin{equation*}
P_{N}(z)\simeq\frac{(-1)^{N}}{2^{\alpha+1}}\cdot\left(\frac{b}{4}\right)^{N}\cdot\left(\sqrt{\eta}+\sqrt{1+\eta}\right)^{2N+\alpha+1}\cdot\frac{\exp\left[-\frac{2N+\alpha}{2\beta}-\mathcal{I}(z)
\right]}{\eta^{\frac{2\alpha+1}{4}}(1+\eta)^{\frac{1}{4}}}.
\end{equation*}

Thus the orthonomal polynomials $\mathcal{P}_{N}(z)$ of Theorem \ref{1111} can be obtained using the standard method, stated as the below Lemma.
\end{proof}

\begin{lem}\label{lem1111111}
{\rm \cite{C20}} The orthonomal polynomials $\mathcal{P}_{N}(z)$ with respect to the weight $w(x)$, i.e.
\begin{equation*}
\int_{a}^{b}\left[\mathcal{P}_{N}(x)\right]^{2}w(x)dx=1,
\end{equation*}
can be given by:
\begin{equation*}
\mathcal{P}_{N}(z)=\sqrt{\frac{2}{\pi(b-a)}}\exp\left[\frac{A}{2}\right]P_{N}(z),
\end{equation*}
where
\begin{equation*}
A:=2\int_{a}^{b}\frac{v(x)dx}{2\pi\sqrt{(b-x)(x-a)}}-2N\log\left(\frac{b-a}{4}\right),
\end{equation*}
and the orthogonal polynomials $P_{N}(z)$ is approximated by (\ref{pn}).
\end{lem}

\begin{rem}\label{rem01}
Apparently, the first representation in (\ref{0.01}) is more convenient for sufficiently large $N$, where $|\eta|\ll1$. However, it cannot be used for $\beta=N+\frac{1}{2},~N=1,2,\ldots$ by the nature of the Hypergeometric function, that is why the second expression in (\ref{0.01}) is needed.
\end{rem}
To make further progress, we will be continuing to simplify the representation of $\mathcal{P}_{N}(z)$. Using the inverse hyperbolic sine and the formula in \cite{new2}  (cf. 9.121. 26), the following identity holds
\begin{equation}\label{100}
\log\left(\sqrt{\eta}+\sqrt{\eta+1}\right)={\rm arcsinh}\sqrt{\eta}=\sqrt{\eta}\cdot{_{2}F_{1}}\left(\frac{1}{2},\frac{1}{2};\frac{3}{2};-\eta\right).
\end{equation}
According to this, if we denote $E[\beta-\frac{1}{2}]$\footnote[1]{ Throughout this paper, $E[x]$ denotes the integer part of $x$.} by $E_{\beta}$, we have
\begin{lem}\label{2222}
The asymptotic expression of the polynomials for $z\notin[0,\infty)$, $|\eta|\ll1$, is,
\begin{equation}\label{0.3}
\begin{split}
\mathcal{P}_{N}(z)\simeq\frac{(-1)^{N}\eta^{\frac{1}{4}}}{\sqrt{2\pi (-z)^{\alpha+1}}}\cdot\exp\left[-\mathcal{I}(z)+
\frac{(-z)^{\beta}}{2\sqrt{\pi}C^{\beta}}\sum_{k=0}^{E_{\beta}}(-1)^{k}\cdot a_{k}\cdot\eta^{k-\beta+\frac{1}{2}}
\right],
\end{split}
\end{equation}
where $\mathcal{I}(z)$ is given in (\ref{0.01}) and $a_{k}$ is defined as
\begin{equation}\label{0.02}
a_{k}:=\frac{\Gamma\left(k+\frac{1}{2}\right)}{\left(k+\frac{1}{2}\right)\Gamma\left(k+1\right)}.
\end{equation}
\end{lem}

\begin{proof}
By (\ref{100}), we find
\begin{equation*}
\begin{split}
\left(2N+\alpha+1\right)&\log\left(\sqrt{\eta}+\sqrt{1+\eta}\right)\simeq\left(2N+\alpha\right)\sqrt{\eta}\cdot{_{2}F_{1}}\left(\frac{1}{2},\frac{1}{2};\frac{3}{2};-\eta\right)\\
&=\frac{(-z)^{\beta}}{\eta^{\beta}}\frac{1}{C^{\beta}}\sqrt{\eta}\cdot\sum_{k=0}^{\infty}\frac{\left(\frac{1}{2}\right)_{k}\left(\frac{1}{2}\right)_{k}}{\left(\frac{3}{2}
\right)_{k}k!}(-\eta)^{k}\\
&\simeq\frac{(-z)^{\beta}}{2\sqrt{\pi}C^{\beta}}\sum_{k=0}^{E_{\beta}}(-1)^{k}\cdot \frac{\Gamma\left(k+\frac{1}{2}\right)}{\left(k+\frac{1}{2}\right)\Gamma\left(k+1\right)}\cdot\eta^{k-\beta+\frac{1}{2}},
\end{split}
\end{equation*}
where, the Pochhammer symbol (also called the shifted factorial) reads
\begin{equation*}
(x)_{k}:=\frac{\Gamma\left(k+x\right)}{\Gamma(x)}=x(x+1)\cdots(x+k-1).
\end{equation*}
Hence the Lemma \ref{2222} is obtained immediately.
\end{proof}

In sections 4 and 5, we will follow the techniques of \cite{C3} and \cite{C9} to show that using an appropriate selection of vectors $\left\{\xi_{m}\right\}$, that the lower bound given by (\ref{e4}) is actually an asymptotic estimate of $\lambda_{N}$ for sufficiently large $N$. Taking full advantage of the Laplace method, we can obtain an estimation of $\sum_{n=0}^{N}\mathcal{K}_{nn}$. Consequently, the asymptotic behavior of $\lambda_{N}$ follows.

As mentioned in the Remark \ref{rem01}, our problem will be discussed in two different cases.
\section{The approximation of $\lambda_{N}$ for $\beta\neq n+\frac{1}{2},{n\in\{1,2,3,\ldots\}}$}
To find the asymptotic estimate of $\sum_{n=0}^{N}\mathcal{K}_{nn}$ for $\beta\neq n+\frac{1}{2},~n\in\{1,2,3,\ldots\}$, we will first deal with the term $\mathcal{I}(z)$ in (\ref{0.3}) by using the first form in equation (\ref{0.01}).

\begin{lem}\label{300}
For $\beta\neq n+\frac{1}{2},~n\in\{1,2,3,\ldots\}$, we have
\begin{equation}\label{0.4}
\begin{split}
\mathcal{P}_{N}(z)&\simeq\frac{(-1)^{N}}{\sqrt{2\pi(-z)^{\alpha+\frac{1}{2}}}}\cdot\left[C\left(2N+\alpha\right)^{\frac{1}{\beta}}\right]^{-\frac{1}{4}}\cdot
\exp\left[\frac{(-z)^{\beta}}{2}\sec\pi\beta\right]\\
&~~~~\cdot\exp\left[\frac{\left(2N+\alpha\right)^{1-\frac{1}{2\beta}}}{2\sqrt{\pi C}}\sum_{k=0}^{E_{\beta}}(-1)^{k}A_{k}\frac{(-z)^{k+\frac{1}{2}}}{\left(C\left(2N+\alpha\right)^{\frac{1}{\beta}}\right)^{k}}\right],
\end{split}
\end{equation}
here
\begin{equation*}
A_{k}:=a_{k}+\frac{\Gamma\left(\frac{1}{2}-\beta\right)}{2\Gamma(1-\beta)}b_{k} ~~~{\rm with} ~~~ b_{k}:=\sum_{j=0}^{k}\frac{\Gamma\left(j-\frac{1}{2}\right)\Gamma\left(k-j+1-\beta\right)}{\Gamma\left(j+1\right)\Gamma\left(k-j+\frac{3}{2}-\beta\right)},
\end{equation*}
where $a_{k}$ is same as that given in (\ref{0.02}).

Specially,
\begin{equation}\label{10}
A_{0}=\frac{4\sqrt{\pi}\beta}{2\beta-1}.
\end{equation}
\end{lem}

\begin{proof}
For $|\eta|<1$, the hypergeometric function ${_{2}F_{1}}\left(1,1-\beta;\frac{3}{2}-\beta; -\eta\right)$ has the below series expansion
\begin{equation*}
{_{2}F_{1}}\left(1,1-\beta;\frac{3}{2}-\beta; -\eta\right)=\frac{\Gamma\left(\frac{3}{2}-\beta\right)}{\Gamma\left(1-\beta\right)}\sum_{k=0}^{\infty}(-1)^{k}\frac{\Gamma\left(k+1-\beta\right)}{\Gamma\left(k+\frac{3}{2}-\beta\right)}
\eta^{k}.
\end{equation*}
Applying the formula (\cite{new2},~{\rm p}1015)
\begin{equation*}
{_{2}F_{1}}\left(-n,\beta;\beta;-z\right)=(1+z)^{n},
\end{equation*}
then for $|\eta|<1$, $\sqrt{1+\eta}$ may be written as
\begin{equation}\label{0.9}
\sqrt{1+\eta}=\frac{1}{\Gamma\left(-\frac{1}{2}\right)}\sum_{k=0}^{\infty}(-1)^{k}\frac{\Gamma\left(k-\frac{1}{2}\right)}{\Gamma\left(k+1\right)}\eta^{k}.
\end{equation}
So as $\eta\rightarrow0$, the expansion for $\mathcal{I}(z)$ is
\begin{equation}\label{500}
\begin{split}
&\mathcal{I}(z)\simeq-\frac{1}{4\sqrt{\pi}}\left(\frac{-z}{C}\right)^{\beta}\frac{\Gamma\left(\frac{1}{2}-\beta\right)}{\Gamma(1-\beta)}\sum_{k=0}^
{E_{\beta}}(-1)^{k}b_{k}\cdot\eta^{k-\beta+\frac{1}{2}}-\frac{(-z)^{\beta}}{2}\sec\pi\beta,\\
\end{split}
\end{equation}
where
\begin{equation*}
b_{k}:=\sum_{j=0}^{k}\frac{\Gamma\left(j-\frac{1}{2}\right)\Gamma\left(k-j+1-\beta\right)}{\Gamma\left(j+1\right)\Gamma\left(k-j+\frac{3}{2}-\beta\right)}.
\end{equation*}
Substituting (\ref{500}) into (\ref{0.3}), and bear in mind $\eta=-zC^{-1}(2N+\alpha)^{-\frac{1}{\beta}}$, then the Lemma \ref{300} follows.
\end{proof}

\begin{rem}\label{400}
Letting $\alpha=0,~ \beta=1$, we find $C=2$ and $A_{0}=4\sqrt{\pi}$. Consequently, the classical result for Laguerre polynomials due to Perron \cite{new3} is recovered,
\begin{equation*}
\mathcal{P}_{N}(z)\simeq\frac{(-1)^{N}}{2\sqrt{\pi}}\left(-zN\right)^{-\frac{1}{4}}\exp\left[\frac{z}{2}+2\sqrt{-zN}\right],~~z\notin[0,\infty).
\end{equation*}
\end{rem}
\begin{rem}\label{20}
The Laplace method \cite{C27} gives,
\begin{equation*}
\int_{a}^{b}f(t){\rm e}^{-\lambda g(t)}dt\simeq {\rm e}^{-\lambda g(c)}f(c)\sqrt{\frac{2\pi}{\lambda g''(c)}}~,~~~~{\rm as}~~\lambda\rightarrow\infty,
\end{equation*}
where $g$ assumes a strict minimum over $[a,b]$ at an interior critical point $c$, such that
\begin{equation*}
\begin{split}
g'(c)=0,~~~~~g''(c)>0~~~~~{\rm and}~~~~~f(c)\neq0.\\
\end{split}
\end{equation*}
An alternative expression for Laplace method may be stated as:

If for $x\in[a,b]$, the real continuous function $g(x)$ has as its maximum the value $g(b)$, then as $N\rightarrow\infty$
\begin{equation}\label{laplace2}
\int_{a}^{b}f(x){\rm e}^{Ng(x)}dx\simeq\frac{f(b){\rm e}^{Ng(b)}}{Ng'(b)}.
\end{equation}
\end{rem}
\begin{thm}\label{th3}
For $\beta\neq n+\frac{1}{2},~n\in\{1,2,3,\ldots\}$, the smallest eigenvalue $\lambda_{N}$ of the $\mathcal{H}_{N}$ can be approximated by
\begin{equation}\label{0.03}
\begin{split}
&\lambda_{N}\simeq\\
&2^{\frac{5}{2}}\pi^{\frac{5}{4}}C^{-\frac{1}{4}}A_{0}^{\frac{1}{2}}\left(2N+\alpha\right)^{\frac{1}{2}-\frac{1}{4\beta}}\exp\left[-\sec\pi\beta-\frac{\left(2N+\alpha\right)^{1-\frac{1}{2\beta}}}
{\sqrt{\pi C}}\sum_{k=0}^{E_{\beta}}(-1)^{k}\frac{A_{k}}{C^{k}}\left(2N+\alpha\right)^{-\frac{k}{\beta}}\right],
\end{split}
\end{equation}
where
\begin{equation*}
C:=4\left[\frac{\Gamma(\beta+1)\Gamma(\beta)}{\Gamma(2\beta+1)}\right]^{\frac{1}{\beta}},~~   A_{k}:=a_{k}+\frac{\Gamma\left(\frac{1}{2}-\beta\right)}{2\Gamma(1-\beta)}b_{k},~~~\beta>\frac{1}{2},
\end{equation*}
with
\begin{equation*}
a_{k}:=\frac{\Gamma\left(k+\frac{1}{2}\right)}{\left(k+\frac{1}{2}\right)\Gamma\left(k+1\right)}
,~~~~b_{k}:=\sum_{j=0}^{k}\frac{\Gamma\left(j-\frac{1}{2}\right)\Gamma\left(k-j+1-\beta\right)}{\Gamma\left(j+1\right)\Gamma\left(k-j+\frac{3}{2}-\beta\right)}
~.
\end{equation*}
\end{thm}

\begin{proof}
Note that $A_{0}>0$ for $\beta>\frac{1}{2}$ by (\ref{10}), so the essential contribution to $\mathcal{K}_{\mu\nu}$ comes from a small neighborhood of $z=-1$ as $\mu\rightarrow\infty$ and $\nu\rightarrow\infty$. Let $\omega>0$ be a fixed number and restrict the values of $\mu$ and $\nu$ to satisfy
\begin{equation}\label{0.5}
N-\omega N^{\frac{1}{2\beta}}\leq\mu,\nu\leq N, ~~N\rightarrow\infty,
\end{equation}
thus we have
\begin{equation}\label{neweq1}
\mathcal{K}_{\mu\nu}\simeq\int_{-\varepsilon}^{\varepsilon}\mathcal{P}_{\mu}\left(-{\rm e}^{{\rm i}\theta}\right)\mathcal{P}_{\nu}\left(-{\rm e}^{-{\rm i}\theta}\right)d\theta.
\end{equation}
Expanding the integrand for $|\theta|\ll1$, we obtain
\begin{equation}\label{30}
\begin{split}
\mathcal{K}_{\mu\nu}\simeq&\frac{(-1)^{\mu+\nu}{\rm e}^{\sec\pi\beta}}{2\pi\sqrt{C}}\cdot\left(2N+\alpha\right)^{-\frac{1}{2\beta}}\cdot\int_{-\varepsilon}^{\varepsilon}\exp{\Bigg[}
\frac{1}{2\sqrt{\pi C}}\sum_{k=0}^{E_{\beta}}(-1)^{k}\frac{A_{k}}{C^{k}}\\
&\cdot\left(1-\frac{(2k+1)^{2}\theta^{2}}{8}\right)
\cdot\left((2\mu+\alpha)^{1-\frac{1}{2\beta}-\frac{k}{\beta}}+(2\nu+\alpha)^{1-\frac{1}{2\beta}-\frac{k}{\beta}}\right)\\
&+\frac{(2k+1){\rm i}\theta}{2}\left((2\mu+\alpha)^{1-\frac{1}{2\beta}-\frac{k}{\beta}}-(2\nu+\alpha)^{1-\frac{1}{2\beta}-\frac{k}{\beta}}\right){\Bigg]}d\theta.
\end{split}
\end{equation}
Note that $(2\mu+\alpha)^{1-\frac{1}{2\beta}-\frac{k}{\beta}}-(2\nu+\alpha)^{1-\frac{1}{2\beta}-\frac{k}{\beta}}$ remains bounded because of restricting $\mu$ and $\nu$ as in (\ref{0.5}),
so we can get rid of the linear term in (\ref{30}) for $\theta\ll1$. As mentioned above, contributions to the integral (\ref{30}) from $(-\infty,\varepsilon)$ and $(\varepsilon,\infty)$ are small enough compared with those from $[-\varepsilon,\varepsilon]$ as $\mu\rightarrow\infty$ and $\nu\rightarrow\infty$. Therefore, we can extend the integration interval to $\mathbb{R}$ but without affecting the approximation of $\mathcal{K}_{\mu\nu}$. Using the Laplace method given by Remark \ref{20}, we obtain
\begin{equation}\label{0.6}
\begin{split}
\mathcal{K}_{\mu\nu}\simeq&\frac{(-1)^{\mu+\nu}}{(\pi C)^{\frac{1}{4}}}\sqrt{\frac{2}{A_{0}}}\left(2N+\alpha\right)^{-\frac{1}{2}-\frac{1}{4\beta}}{\rm e}^{\sec\pi\beta}\\
&\cdot\exp\left[\frac{1}{2\sqrt{\pi C}}\sum_{k=0}^{E_{\beta}}(-1)^{k}\frac{A_{k}}{C^{k}}
\left((2\mu+\alpha)^{1-\frac{1}{2\beta}-\frac{k}{\beta}}+(2\nu+\alpha)^{1-\frac{1}{2\beta}-\frac{k}{\beta}}\right)\right].
\end{split}
\end{equation}

Observing (\ref{0.6}), we can find that when $\mu$ and $\nu$ satisfied (\ref{0.5}) and large enough,
\begin{equation}\label{0.7}
\mathcal{K}_{\mu\nu}\simeq(-1)^{\mu+\nu}\mathcal{K}_{\mu\mu}^{\frac{1}{2}}\mathcal{K}_{\nu\nu}^{\frac{1}{2}}.
\end{equation}

Using the approach of \cite{C3} and \cite{C9} with the following choices of $\left\{\xi_{\nu}\right\}$, allows us to determine the asymptotic behavior of $\lambda_{N}$ for large $N$,
\begin{equation*}
\xi_{\nu}=
\begin{cases}
(-1)^{\nu}\sigma\mathcal{K}_{\nu\nu}^{\frac{1}{2}}, ~{\rm if} ~N_{0}\leq\nu\leq N,\ \ \\\\
0,\ \
~~~~~~~~~~~~{\rm if}~\nu<N_{0}:=E\left[N-\omega N^{\frac{1}{2\beta}}\right],
\end{cases}
\end{equation*}
and the positive number $\sigma$ is determined by the condition
\begin{equation}\label{0.8}
\sum_{\nu=0}^{N}|\xi_{\nu}|^{2}=\sigma^{2}\sum_{\nu=N_{0}}^{N}\mathcal{K}_{\nu\nu}=1.
\end{equation}
It follows from (\ref{0.7}) and (\ref{0.8}) that
\begin{equation}\label{40}
\begin{split}
\sum_{\mu,\nu=0}^{N}\mathcal{K}_{\mu\nu}\xi_{\mu}\overline{\xi}_{\nu}=\sum_{\mu,\nu=N_{0}}^{N}(-1)^{\mu+\nu}\sigma^{2}\mathcal{K}_{\mu\nu}\mathcal{K}_{\mu\mu}^{\frac{1}{2}}\mathcal{K}_{\nu\nu}^{\frac{1}{2}}
\simeq\sigma^{2}\left(\sum_{\nu=N_{0}}^{N}\mathcal{K}_{\nu\nu}\right)^{2}
=\sum_{\nu=N_{0}}^{N}\mathcal{K}_{\nu\nu}.
\end{split}
\end{equation}
This means the minimum value in equation (\ref{e3}) can be approximated by (\ref{e4}), following (\ref{40}), because of the arbitrariness of $\omega$, i.e.
\begin{equation*}
\lambda_{N}\simeq\frac{2\pi}{\sum_{\nu=0}^{N}\mathcal{K}_{\nu\nu}}.
\end{equation*}
It follows that
\begin{equation}\label{5000}
\lambda_{N}\simeq\frac{2\pi}{\int_{0}^{N}\mathcal{K}_{\nu\nu}d\nu}.
\end{equation}
Substituting (\ref{0.6}) into (\ref{5000}), with a simple calculation by applying the Laplace method, see Remark \ref{20}, then the asymptotic behavior of $\lambda_{N}$, for $\beta\neq n+\frac{1}{2},n\in\{1,2,3,\ldots\}$, is obtained.
\end{proof}
\begin{exmp}\label{40000}
If we take $\alpha=-\frac{1}{2}$, $\beta=\frac{7}{4}$, then
\begin{equation*}
\lambda_{N}\simeq2^{\frac{5}{2}}\pi^{\frac{5}{4}}C^{-\frac{1}{4}}A_{0}^{\frac{1}{2}}\left(2N-\frac{1}{2}\right)^{\frac{5}{14}}
\exp\left[-\sqrt{2}-\frac{\left(2N-\frac{1}{2}\right)^{\frac{5}{7}}}{\sqrt{\pi C}}\left(A_{0}-\frac{A_{1}}{C}\left(2N-\frac{1}{2}\right)^{-\frac{4}{7}}\right)\right],
\end{equation*}
where
$C=4\left[\frac{\Gamma\left(\frac{11}{4}\right)\Gamma\left(\frac{7}{4}\right)}{\Gamma\left(\frac{9}{2}\right)}\right]^{\frac{4}{7}}$, $A_{0}=\frac{14\sqrt{\pi}}{5}$ and $A_{1}=\frac{7\sqrt{\pi}}{3}$.
\end{exmp}
\begin{cor}\label{cor1}
For the classical Laguerre weight $x^{\alpha}{\rm e}^{-x},~x\in[0,\infty),$
$\alpha>-1$, {\rm i.e.} taking $\beta=1$ for our weight $w(x)$, we have
\begin{equation*}
\lambda_{N}\simeq2^{\frac{13}{4}}\pi^{\frac{3}{2}}{\rm e}\left(2N+\alpha\right)^{\frac{1}{4}}\exp\left[-2^{\frac{3}{2}}\left(2N+\alpha\right)^{\frac{1}{2}}\right].
\end{equation*}
\end{cor}

\begin{rem}
When $\alpha=0,~\beta=1$, Szeg\"{o}'s \cite{C3}\footnote{The original formula of $\lambda_{N}$ in the last equation on page 461 missed a factor of 4.} classical result for the Laguerre weight ${\rm e}^{-x}$ is recovered:
\begin{equation*}
\begin{split}
&\lambda_{N}\simeq2^{\frac{7}{2}}\pi^{\frac{3}{2}}{\rm e}N^{\frac{1}{4}}\exp\left[-4N^{\frac{1}{2}}\right].
\end{split}
\end{equation*}
\end{rem}
\begin{rem}
With the restriction of $\alpha=0, \beta\neq n+\frac{1}{2}, n\in\{1,2,3,\ldots\}$, Chen and Lawrence's result on the weight ${\rm e}^{-x^{\beta}},~x\in[0,\infty)$ is also recovered:
\begin{equation*}
\lambda_{N}\simeq
8\pi^{\frac{5}{4}}\widetilde{C}^{-\frac{1}{4}}\widetilde{A}_{0}^{\frac{1}{2}}N^{\frac{1}{2}-\frac{1}{4\beta}}\exp\left[-\sec\pi\beta-\frac{N^{1-\frac{1}{2\beta}}}
{\sqrt{\pi \widetilde{C}}}\sum_{r=0}^{E_{\beta}}(-1)^{r}\frac{\widetilde{A}_{r}}{\widetilde{C}^{r}}N^{-\frac{r}{\beta}}\right],
\end{equation*}
for details of $\widetilde{A}_{0},\widetilde{A}_{r},\widetilde{C}$, please see \cite{C9}.
\end{rem}

From (\ref{0.03}) we find that $\lambda_{N}$ is exponentially small for large $N$ and tends to $0$ as $N\rightarrow\infty$.

\section{ The approximation of $\lambda_{N}$ for $\beta=n+\frac{1}{2},{n\in\{1,2,3,\ldots\}}$}
Our goal for this section is to find the approximation of $\lambda_{N}$ for the cases where $\beta=n+\frac{1}{2},n\in\{1,2,3,\ldots\}$. Such cases, as was illustrated in Remark \ref{rem01}, require the second representation of $\mathcal{I}(z)$ in (\ref{0.01}). Before obtaining the asmptotic behavior of $\lambda_{N}$, we first establish the following lemma for $\mathcal{P}(z)$.

\begin{lem}\label{lem0.1}
For $|\eta|\ll1$, then as $N\rightarrow\infty$,
\begin{equation*}
\mathcal{I}(z)\simeq\frac{(-1)^{\beta+\frac{1}{2}}(-z)^{\beta}}{2\pi}\log\left[\frac{4}{\eta}\right]+\frac{(-z)^{\beta}}{4\pi^{\frac{3}{2}}}\sum_{k=0}^{\beta-\frac{1}{2}}(-1)^{k}
\delta_{\beta-\frac{1}{2}-k}\eta^{k-\beta+\frac{1}{2}}
\end{equation*}
where
\begin{equation*}
L_{k}:=\frac{k}{\pi}C^{k}(k),~~~{\rm and}~~~\delta_{k}:=\sum_{j=1}^{\beta-\frac{1}{2}}\frac{\gamma_{j-k}}{L_{j-\frac{1}{2}}},
\end{equation*}
with
\begin{equation*}
\gamma_{k}:=
\begin{cases}
\frac{\Gamma\left(k-\frac{1}{2}\right)}{\Gamma(k+1)}, ~{\rm if} ~~k\geq0, \\\\
0,\ \
~~~~~~ {\rm if}~~k<0.
\end{cases}
\end{equation*}
\end{lem}

\begin{proof}
Based on the Gauss' recursion relation \cite{new2}, see (\ref{0.0}) in the Appendix, Chen and Lawrence \cite{C9} built the following version formula:
\begin{equation*}
\begin{split}
{_{2}F_{1}}\left(1,\frac{1}{2};n+\frac{5}{2};z\right)&=\frac{\left(n+\frac{3}{2}\right)(z-1)}{(n+1)z}\left[{_{2}F_{1}}\left(1,\frac{1}{2};n+\frac{3}{2};z\right)
-{_{2}F_{1}}\left(1,\frac{1}{2};n+\frac{1}{2};z\right)\right]\\
&+\frac{n\left(n+\frac{3}{2}\right)}{(n+1)\left(n+\frac{1}{2}\right)}{_{2}F_{1}}\left(1,\frac{1}{2};n+\frac{3}{2};z\right),
\end{split}
\end{equation*}
together with the fact that
\begin{equation*}
{_{2}F_{1}}\left(1,\frac{1}{2};\frac{5}{2};z\right)=\frac{3}{4}\frac{(z-1)}{z^{\frac{3}{2}}}\log\left[\frac{1+\sqrt{z}}{1-\sqrt{z}}\right]+\frac{3}{2}z,
\end{equation*}
we can get
\begin{equation*}
{_{2}F_{1}\left(1,\frac{1}{2};\beta+1;z\right)}=L_{\beta}\frac{(z-1)^{\beta-\frac{1}{2}}}{z^{\beta+\frac{1}{2}}}\left(\sqrt{z}\log\left[\frac{1+\sqrt{z}}{1-\sqrt{z}}\right]+
\sum_{k=1}^{\beta-\frac{1}{2}}\frac{1}{L_{k-\frac{1}{2}}}\left(\frac{z}{z-1}\right)^{k}\right),
\end{equation*}
where $\beta=n+\frac{1}{2},n=1,2,3,\ldots$ and $L_{k}$ is given by
\begin{equation*}
L_{k}:=\frac{k}{\pi}C^{k}(k).
\end{equation*}

\noindent Consequently, we have
\begin{equation*}
\mathcal{I}(z)=\frac{(-1)^{\beta+\frac{1}{2}}}{2\pi}(-z)^{\beta}\left(\log\left[\frac{\sqrt{1+\eta}+1}{\sqrt{1+\eta}-1}\right]+\sqrt{1+\eta}
\sum_{k=1}^{\beta-\frac{1}{2}}(-1)^{k}\frac{\eta^{-k}}{L_{k-\frac{1}{2}}}\right).
\end{equation*}
By using (\ref{0.9}) again, we find
\begin{equation*}
\begin{split}
\mathcal{I}(z)&\simeq\frac{(-1)^{\beta+\frac{1}{2}}(-z)^{\beta}}{2\pi}\log\left[\frac{4}{\eta}\right]+\frac{(-1)^{\beta-\frac{1}{2}}(-z)^{\beta}}{4\pi^{\frac{3}{2}}}
\left[\sum_{j=1}^{\beta-\frac{1}{2}}\frac{(-1)^{j}\eta^{-j}}{L_{j-\frac{1}{2}}}\right]\\
&\cdot\left[\sum_{k=0}^{\beta-\frac{1}{2}}(-1)^{k+j-\left(\beta-\frac{1}{2}\right)}
\frac{\Gamma\left(k+j-\left(\beta-\frac{1}{2}\right)-\frac{1}{2}\right)}{\Gamma\left(k+j-\left(\beta-\frac{1}{2}\right)+1\right)}\eta^{k+j-\left(\beta-\frac{1}{2}\right)}\right]\\
&=\frac{(-1)^{\beta+\frac{1}{2}}(-z)^{\beta}}{2\pi}\log\left[\frac{4}{\eta}\right]+
\left[\sum_{k=0}^{\beta-\frac{1}{2}}(-1)^{k}\sum_{j=1}^{\beta-\frac{1}{2}}
\frac{\frac{\Gamma\left(k+j-\left(\beta-\frac{1}{2}\right)-\frac{1}{2}\right)}{\Gamma\left(k+j-\left(\beta-\frac{1}{2}\right)+1\right)}}{L_{j-\frac{1}{2}}}
\eta^{k+\frac{1}{2}-\beta}\right].\\
\end{split}
\end{equation*}
With an easy simplification, the Lemma is obtained immediately.
\end{proof}

Substituting $\eta=-zC^{-1}(2N+\alpha)^{-\frac{1}{\beta}}$, together with a simple calculation gives the following strong asymptotics of $\mathcal{P}_{N}(z)$ for $z\notin[0,\infty)$,
\begin{lem}\label{lem0.2}
For $z\notin[0,\infty)$, we have
\begin{equation*}
\begin{split}
\mathcal{P}_{N}(z)&\simeq\frac{(-1)^{N}}{\sqrt{2\pi}}(-z)^{-\frac{\alpha}{2}-\frac{1}{4}}\left(C(2N+\alpha)^{\frac{1}{\beta}}\right)^{-\frac{1}{4}}
\exp\left[\frac{(-1)^{\beta-\frac{1}{2}}(-z)^{\beta}}{2\pi}
\log\left(\frac{4C(2N+\alpha)^{\frac{1}{\beta}}}{-z}\right)\right]\\
&\cdot\exp\left[\frac{(2N+\alpha)^{1-\frac{1}{2\beta}}}{2\sqrt{\pi C}}\sum_{k=0}^{\beta-\frac{1}{2}}(-1)^{k}B_{k}\frac{(-z)^{k+\frac{1}{2}}}{\left(C(2N+\alpha)^{\frac{1}{\beta}}\right)^{k}}\right],\\
\end{split}
\end{equation*}
where
\begin{equation*}
B_{k}:=a_{k}-\frac{L_{\beta}}{2\beta}\delta_{\beta-\frac{1}{2}-k},
\end{equation*}
In particularly,
\begin{equation*}
B_{0}=\frac{4\sqrt{\pi}\beta}{2\beta-1}.
\end{equation*}
\end{lem}

\begin{thm}\label{th0.3}
For $\lambda_{N}$, we have
\begin{equation}\label{e5}
\begin{split}
\lambda_{N}\simeq&~2^{\frac{5}{2}}\pi^{\frac{5}{4}}C^{-\frac{1}{4}}B_{0}^{\frac{1}{2}}\left(2N+\alpha\right)^{\frac{1}{2}-\frac{1}{4\beta}}
\left(4C\left(2N+\alpha\right)^{\frac{1}{\beta}}\right)^{\frac{(-1)^{\beta+\frac{1}{2}}}{\pi}}\\
&\cdot\exp\left[-\frac{(2N+\alpha)^{1-\frac{1}{2\beta}}}{\sqrt{\pi C}}\sum_{k=0}^{\beta-\frac{1}{2}}(-1)^{k}\frac{B_{k}}{C^{k}}\left(2N+\alpha\right)^{-\frac{k}{\beta}}\right].
\end{split}
\end{equation}
\end{thm}

\begin{proof}
Since $B_{0}>0$ and by an argument like that in the Section 4, again we find that the dominant contribution to $\mathcal{K}_{\mu\nu}$ is from the arc of the unit circle around $z=-1$. Restricting $\mu,~\nu$ to the same range given by (\ref{0.5}), then $(2\mu+\alpha)^{1-\frac{1}{2\beta}-\frac{k}{\beta}}-(2\nu+\alpha)^{1-\frac{1}{2\beta}-\frac{k}{\beta}}$ and $\log\left[(2\mu+\alpha)/(2\nu+\alpha)\right]$ remain bounded and (\ref{neweq1}) will also be true at here. By the Laplace method, we have
\begin{equation}\label{neweq2}
\mathcal{K}_{\mu\nu}\simeq\int_{-\infty}^{\infty}\mathcal{P}_{\mu}\left(-{\rm e}^{{\rm i}\theta}\right)\mathcal{P}_{\nu}\left(-{\rm e}^{-{\rm i}\theta}\right)d\theta.
\end{equation}

As previously, we expand the exponential in the integrand for $|\theta|\ll1$, reserving terms up to the second order.  We obtain\
\begin{equation*}
\begin{split}
\mathcal{K}_{\mu\nu}&\simeq\frac{(-1)^{\mu+\nu}}{(\pi C)^{\frac{1}{4}}}\sqrt{\frac{2}{B_{0}}}\left(2N+\alpha\right)^{-\frac{1}{2}-\frac{1}{4\beta}}\left(4C\left(2N+\alpha\right)^{\frac{1}{\beta}}\right)^{\frac{(-1)^{\beta-\frac{1}{2}}}{\pi}}\\
&\cdot\exp\left[\frac{1}{2\sqrt{\pi C}}\sum_{k=0}^{\beta-\frac{1}{2}}(-1)^{k}\frac{B_{k}}{C^{k}}\left((2\mu+\alpha)^{1-\frac{1}{2\beta}-\frac{k}{\beta}}+(2\nu+\alpha)^{1-\frac{1}{2\beta}-\frac{k}{\beta}}\right)\right],
\end{split}
\end{equation*}

For large enough $\mu$ and $\nu$, restricted by (\ref{0.5}), again we will have
\begin{equation*}
\mathcal{K}_{\mu\nu}\simeq(-1)^{\mu+\nu}\mathcal{K}_{\mu\mu}^{\frac{1}{2}}\mathcal{K}_{\nu\nu}^{\frac{1}{2}}.
\end{equation*}
As per the discussion in the previous section, it follows that
\begin{equation*}
\lambda_{N}\simeq\frac{2\pi}{\int_{0}^{N}\mathcal{K}_{\nu\nu}d\nu}.
\end{equation*}
Taking an application of the Laplace method and doing the same argument as before, we get the asymptotic behavior for the integration,
\begin{equation}\label{neweq3}
\begin{split}
\int_{0}^{N}\mathcal{K}_{\nu\nu}d\nu\simeq&~2^{-\frac{3}{2}}\pi^{-\frac{1}{4}}C^{\frac{1}{4}}B_{0}^{-\frac{1}{2}}\left(2N+\alpha\right)^{-\frac{1}{2}+\frac{1}{4\beta}}
\left(4C\left(2N+\alpha\right)^{\frac{1}{\beta}}\right)^{\frac{(-1)^{\beta-\frac{1}{2}}}{\pi}}\\
&\cdot\exp\left[\frac{(2N+\alpha)^{1-\frac{1}{2\beta}}}{\sqrt{\pi C}}\sum_{k=0}^{\beta-\frac{1}{2}}(-1)^{k}\frac{B_{k}}{C^{k}}\left(2N+\alpha\right)^{-\frac{k}{\beta}}\right],
\end{split}
\end{equation}
which completes the proof of this theorem.
\end{proof}

\begin{exmp}
If we take $\alpha=-\frac{3}{4}$, $\beta=\frac{3}{2}$, then
\begin{equation*}
\lambda_{N}\simeq2^{\frac{5}{2}+\frac{2}{\pi}}\pi^{\frac{5}{4}}C^{\frac{1}{\pi}-\frac{1}{4}}B_{0}^{\frac{1}{2}}
\left(2N-\frac{3}{4}\right)^{\frac{1}{3}+\frac{2}{3\pi}}
\exp\left[-\frac{\left(2N-\frac{3}{4}\right)^{\frac{2}{3}}}{\sqrt{\pi C}}\left(B_{0}-\frac{B_{1}}{C}\left(2N-\frac{3}{4}\right)^{-\frac{2}{3}}\right)\right]
\end{equation*}
where $C=\left(\frac{\pi}{2}\right)^{\frac{2}{3}}$, $B_{0}=3\sqrt{\pi}$ and $B_{1}=-\frac{\sqrt{\pi}}{6}$.
\end{exmp}

\begin{rem}
Putting $\alpha=0$, Chen and Lawrence's result for $\beta=n+\frac{1}{2},n=1,2,3,\ldots$ is recovered.
\begin{equation*}
\begin{split}
\lambda_{N}\simeq8\pi^{\frac{1}{4}}\widetilde{C}^{-\frac{1}{4}}\widetilde{B}_{0}^{\frac{1}{2}}N^{\frac{1}{2}-\frac{1}{4\beta}}
\left(4\widetilde{C}N^{\frac{1}{\beta}}\right)^{\frac{(-1)^{\beta+\frac{1}{2}}}{\pi}}
\exp\left[-\frac{2N^{1-\frac{1}{2\beta}}}{\sqrt{\pi \widetilde{C}}}\sum_{k=0}^{\beta-\frac{1}{2}}(-1)^{k}\frac{\widetilde{B}_{k}}{\widetilde{C}^{k}}N^{-\frac{k}{\beta}}\right],
\end{split}
\end{equation*}
where
\begin{equation*}
\widetilde{C}=2^{-\frac{1}{\beta}}C~~~{\rm and}~~~\widetilde{B}_{k}=B_{k}.
\end{equation*}
\end{rem}

Comparing (\ref{0.03}) with (\ref{e5}), we note that the essential difference between them is the term $\exp(-\sec\pi\beta)$ becomes $\left[4C\left(2N+\alpha\right)^{\frac{1}{\beta}}\right]^{\frac{(-1)^{\beta+\frac{1}{2}}}{\pi}}$. The alternating behavior of the second term depends on whether $\beta+\frac{1}{2}$ is even or odd. Anyway, $\lambda_{N}\rightarrow0$ as $N\rightarrow\infty$.

On the basis of the standard theory \cite{C16}, the moment problem with respect to $w(x),~x\in[0,\infty)$ is indeterminate if
\begin{equation*}
\int_{0}^{\infty}\frac{\log w(x) }{\sqrt{x}(1+x)}dx>-\infty.
\end{equation*}
For our weight $w(x)=x^{\alpha}{\rm e}^{-x^{\beta}}, x\in[0,\infty)$,
\begin{equation*}
\int_{0}^{\infty}\frac{-\alpha\log x+x^{\beta} }{\sqrt{x}(1+x)}dx=\int_{0}^{\infty}\frac{x^{\beta} }{\sqrt{x}(1+x)}dx=\pi\sec\left(\pi\beta\right),
\end{equation*}
Therefore, $\beta=\frac{1}{2}$ is the critical point at which the moment problem becomes indeterminate. If we assume the approximation of $\mathcal{P}_{N}(z)$ given in (\ref{0.3}) holds, we see that
\begin{equation}\label{0.0000}
\mathcal{P}_{N}(z)\simeq\frac{(-1)^{N}}{\sqrt{2}\pi}(-z)^{-\frac{\alpha}{2}-\frac{1}{4}}\left(2N+\alpha\right)^{-\frac{1}{2}}\exp\left[\frac{\sqrt{-z}}{\pi}\left(\log\left[\frac{2\pi\left(2N+\alpha
\right)}{\sqrt{-z}}\right]+1\right)\right].
\end{equation}

We assume that both $\mu$ and $\nu$ are large, however $\mu-\nu$ is bounded by a constant, and thus the asymptotic expression holds. Again, we see that the main contributions to $\mathcal{K}_{\mu\nu}$ come from the arc of the unit circle around $z=-1$. However, for $|\eta|\ll1$, it follows the behavior of $\mathcal{P}_{N}$ given by (\ref{0.0000}):
\begin{equation*}
\begin{split}
\mathcal{K}_{\mu\nu}&\simeq\int_{-\varepsilon}^{\varepsilon}\mathcal{P}_{\mu}(-{\rm e}^{{\rm i}\theta})\mathcal{P}_{\nu}(-{\rm e}^{-{\rm i}\theta})d\theta\\
&\simeq\frac{(-1)^{\mu+\nu}}{2\pi^{2}}{\rm e}^{\frac{2}{\pi}+\frac{2}{\pi}\log2\pi}(2\mu+\alpha)^{-\frac{1}{2}+\frac{1}{\pi}}(2\nu+\alpha)^{-\frac{1}{2}+\frac{1}{\pi}}
   \\
&~\cdot\int_{-\varepsilon}^{\varepsilon}\left[1+\frac{2-2\log2\pi-\log(2\mu+\alpha)-\log(2\nu+\alpha)}{8\pi}\theta^{2}\right]d\theta.
\end{split}
\end{equation*}
Quite obviously, $|\mathcal{K}_{\mu\nu}|$ decreases as $\mu$ and $\nu$ increase, which invalidates the argument of the previous section. However, it is possible to get
an approximative lower bound for the smallest eigenvalue using (\ref{e4}).

Using the Christoffel-Darboux formula (\cite{zzz}, Theorem 2.2.2), which reads:
\begin{equation*}
\sum_{j=0}^{N-1}\frac{P_{j}(x)P_{j}(y)}{h_{j}}=\frac{P_{N}(x)P_{N-1}(y)-P_{N}(y)P_{N-1}(x)}{h_{N-1}(x-y)},
\end{equation*}
and is valid for monic orthogonal polynomials $\left\{P_{n}(z)\right\}$, where $h_{j}$ is the square of the $L^{2}$ norm of $P_{N}(z)$, and the result presented in \cite{C23} for large $N$ off-diagonal recurrence coefficients, we have
\begin{equation*}
\begin{split}
\sum_{\nu=0}^{N}\mathcal{K}&_{\nu\nu}=\int_{-\pi}^{\pi}\sum_{\nu=0}^{N}\mathcal{P}_{\nu}\left(-{\rm e}^{{\rm i}\theta}\right)\mathcal{P}_{\nu}\left(-{\rm e}^{-{\rm i}\theta}\right)d\theta\\
&\simeq\frac{\pi^{2}\left(2N+\alpha\right)^{2}}{4}\int_{-\pi}^{\pi}\frac{\mathcal{P}_{N}\left(-{\rm e}^{{\rm i}\theta}\right)\mathcal{P}_{N+1}\left(-{\rm e}^{-{\rm i}\theta}\right)-\mathcal{P}_{N}\left(-{\rm e}^{-{\rm i}\theta}\right)\mathcal{P}_{N+1}\left(-{\rm e}^{{\rm i}\theta}\right)}{{\rm e}^{{\rm i}\theta}-{\rm e}^{-{\rm i}\theta}}d\theta.
\end{split}
\end{equation*}
As a result, applying the Laplace method, gives
\begin{equation*}
\sum_{\nu=0}^{N}\mathcal{K}_{\nu\nu}\simeq\frac{\left[2\pi(2N+\alpha){\rm e}\right]^{\frac{2}{\pi}}}{4\sqrt{\log\left[2\pi(2N+\alpha){\rm e}\right]}}.
\end{equation*}
We found the smallest eigenvalue for $\beta=\frac{1}{2}$ decreases algebraically rather than exponentially, since $\lambda_{N}\simeq\frac{2\pi}{\sum_{\nu=0}^{N}\mathcal{K}_{\nu\nu}}$.

\section{Numerical results}
It is well known that Hankel matrices (moment matrices) of this form are extremely ill-conditioned.  This can be observed
directly from the terms on the main diagonal of $\mathcal{H}_N$ as follows.  The condition number, $\kappa(\mathcal{H}_{N})=\frac{\Lambda_{N}}{\lambda_{N}}$
where $\Lambda_N$ and $\lambda_N$ are the largest and smallest eigenvalues of $\mathcal{H}_N$ respectively.   By applying the
Rayleigh quotient, we know $\Lambda_N$ is greater than all elements on main diagonal and $\lambda_N$ is less than
all elements on the main diagonal. Thus, for some constant $c>0$:
\begin{align*}
   \kappa(\mathcal{H}_{N})  \geq\frac{\mu_{2N}}{\mu_0}
                           = \frac{\Gamma(\frac{\alpha+1+2N}{\beta})}{\Gamma(\frac{\alpha+1}{\beta})}
                            > c\ \Gamma\Big(\frac{2N}{\beta}\Big), ~~~\text{if~~$\frac{2N}{\beta}\geq\frac{3}{2}$}.
\end{align*}
We note that the condition number is of order $\Gamma\left(\frac{2N}{\beta}\right)$. Due to the ill-conditioned nature of these matrices, standard eigensolver packages based double
precision floating values can solve only small instances, i.e. $N<20$, before they exhaust the available precision (53 bit in the
mantissa, 11 bits in the exponent).

In \cite{C11}, Emmart, Chen and Weems developed an efficient parallel algorithm based on arbitrary precision arithmetic
and the Secant method that can handle the extreme ill-conditioning and we employ their algorithms here for our numerical computations. We use the numerical results to test the convergence of our asymptotic formulas to the actual smallest
eigenvalues for various $N$ and several values of the parameters $\alpha$ and $\beta$. Even with efficient software, the
computation times for the largest size, $N=1000$, require almost 10 hours of CPU time on a modern Core i7 processor.

\begin{figure}[H]
\centering
\includegraphics[width=1.1\textwidth]{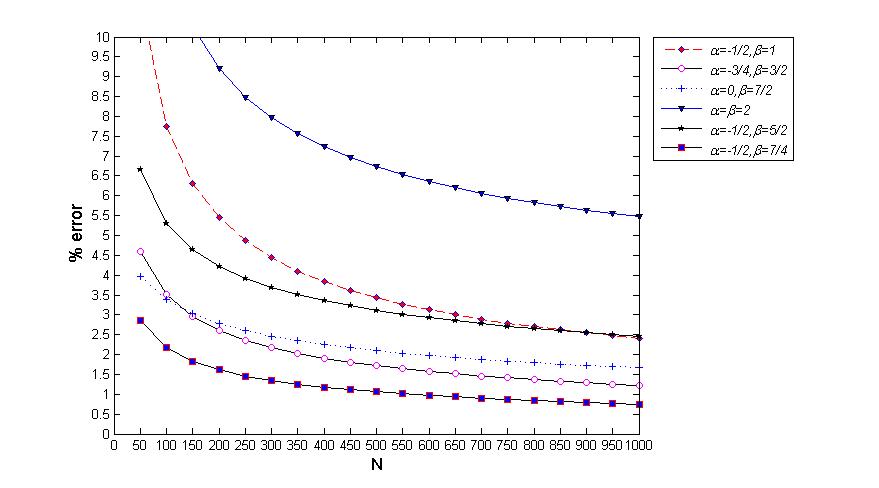}
\caption{The percentage error of the theoretical values of $\lambda_{N}$ vs. those obtained
numerically, for various $\alpha$ and $\beta$. }
\end{figure}

\begin{table}[H]
\centering
\caption{List of numerical results vs. theoretical values, for $\alpha=0,~\beta=\frac{1}{2}$.}
\begin{tabular}{|cllr|cllr|}
\hline
Size $N$      & Numerical $\lambda_{N}$ & Theoretical $\lambda_{N}$    & error     \\  \hline
100           &$0.27397               $ & $0.40360                   $ & $47.32\%$\\
300           &$0.15837               $ & $0.21365                   $ & $34.91\%$\\
500           &$0.12047               $ & $0.15855                   $ & $31.61\%$\\
1000          &$0.082087              $ & $0.10555                   $ & $28.58\%$\\
1500          &$0.065295              $ & $0.083130                  $ & $27.33\%$\\
2000          &$0.055431              $ & $0.070108                  $ & $26.48\%$\\
2500          &$0.048788              $ & $0.061430                  $ & $25.91\%$\\
3000          &$0.043940              $ & $0.055135                  $ & $25.48\%$\\
\hline
\end{tabular}
\end{table}
\begin{rem}\label{rem1}
In Figure 1 and Tables 1-3,
\begin{equation}\label{2.29}
\%~ error:=\left|\frac{{\rm Theoretical}~\lambda_{N}-~{\rm Numerical}~\lambda_{N}}{{\rm Numerical}~\lambda_{N}}\times100\right| .
\end{equation}
\end{rem}

\begin{table}[H]
\centering
\caption{List of numerical results vs. theoretical values, for $\beta\neq n+\frac{1}{2},n\in\mathbb{N}^{+}$}
\begin{tabular}{|cccllr|cccllr|}
\hline
$\alpha$        &$\beta$              &Size $N$      & Numerical $\lambda_{N}$ & Theoretical $\lambda_{N}$    & error     \\  \hline
$ -\frac{1}{2} $&$         1         $&  100         &$2.2434\times10^{-15}     $&$2.4171\times10^{-15}      $& $0.07745$\\
                &                     &  200         &$1.7127\times10^{-22}     $&$1.8059\times10^{-22}      $& $0.05445$\\
                &                     &  300         &$5.7241\times10^{-28}     $&$5.9779\times10^{-28}      $& $0.04435$\\
                &                     &  400         &$1.3647\times10^{-32}     $&$1.4170\times10^{-32}      $& $0.03835$\\
                &                     &  500         &$1.1453\times10^{-36}     $&$1.1845\times10^{-36}      $& $0.03427$\\
                &                     &  600         &$2.3530\times10^{-40}     $&$2.4265\times10^{-40}      $& $0.03126$\\
                &                     &  700         &$9.5340\times10^{-44}     $&$9.8098\times10^{-44}      $& $0.02893$\\
                &                     &  800         &$6.6167\times10^{-47}     $&$6.7956\times10^{-47}      $& $0.02705$\\
                &                     &  900         &$7.1306\times10^{-50}     $&$7.3123\times10^{-50}      $& $0.02549$\\
                &                     & 1000         &$1.1110\times10^{-52}     $&$1.1378\times10^{-52}      $& $0.02417$\\
$      1       $&$        1          $&  100         &$2.0119\times10^{-15}      $&$2.0845\times10^{-15}      $& $0.03610$\\
                &                     &  200         &$1.5845\times10^{-22}     $&$1.6258\times10^{-22}      $& $0.02605$\\
                &                     &  300         &$5.3703\times10^{-28}     $&$5.4855\times10^{-28}      $& $0.02145$\\
                &                     &  400         &$1.2911\times10^{-32}     $&$1.3153\times10^{-32}      $& $0.01867$\\
                &                     &  500         &$1.0898\times10^{-36}     $&$1.1081\times10^{-36}      $& $0.01676$\\
                &                     &  600         &$2.2486\times10^{-40}      $&$2.2831\times10^{-40}     $& $0.01534$\\
                &                     &  700         &$9.1416\times10^{-44}      $&$9.2716\times10^{-44}     $& $0.01423$\\
                &                     &  800         &$6.3614\times10^{-47}       $&$6.4462\times10^{-47}     $&$0.01333$\\
                &                     &  900         &$6.8707\times10^{-50}       $&$6.9572\times10^{-50}     $& $0.01258$\\
                &                     & 1000         &$1.0725\times10^{-52}       $&$1.0853\times10^{-52}     $& $0.01195$\\
$  -\frac{1}{2}$&$  \frac{7}{4}      $&  100         &$2.0753\times10^{-45}     $&$2.1203\times10^{-45}      $& $0.02168$\\
                &                     &  200         &$4.6281\times10^{-76}     $&$4.7027\times10^{-76}      $& $0.01613$\\
                &                     &  300         &$1.7181\times10^{-102}     $&$1.7412\times10^{-102}      $& $0.01344$\\
                &                     &  400         &$1.6945\times10^{-126}     $&$1.7144\times10^{-126}      $& $0.01177$\\
                &                     &  500         &$7.6149\times10^{-149}     $&$7.6955\times10^{-149}      $& $0.01059$\\
                &                     &  600         &$5.9500\times10^{-170}     $&$6.0077\times10^{-170}      $& $0.00970$\\
                &                     &  700         &$4.4336\times10^{-190}     $&$4.4735\times10^{-190}      $& $0.00899$\\
                &                     &  800         &$2.0973\times10^{-209}     $&$2.1149\times10^{-209}      $& $0.00842$\\
                &                     &  900         &$4.7024\times10^{-228}     $&$4.7398\times10^{-228}      $&$0.00794$\\
                &                     & 1000         &$4.0152\times10^{-246}      $&$4.0455\times10^{-246}      $&$0.00753$\\
$      2       $&$        2          $&  100         &$1.5626\times10^{-54}     $&$1.3738\times10^{-54}        $& $0.12082$\\
                &                     &  200         &$4.0862\times10^{-93}     $&$3.7101\times10^{-93}        $& $0.09204$\\
                &                     &  300         &$4.6575\times10^{-127}     $&$4.2866\times10^{-127}      $& $0.07964$\\
                &                     &  400         &$2.7728\times10^{-158}     $&$2.5723\times10^{-158}      $& $0.07230$\\
                &                     &  500         &$1.2618\times10^{-187}     $&$1.1769\times10^{-187}      $& $0.06729$\\
                &                     &  600         &$1.5155\times10^{-215}     $&$1.4191\times10^{-215}      $& $0.06358$\\
                &                     &  700         &$2.4610\times10^{-242}     $&$2.3117\times10^{-242}      $& $0.06067$\\
                &                     &  800         &$3.4223\times10^{-268}     $&$3.2228\times10^{-268}      $& $0.05831$\\
                &                     &  900         &$2.9306\times10^{-293}     $&$2.7654\times10^{-293}      $& $0.05635$\\
                &                     & 1000         &$1.2053\times10^{-317}      $&$1.1394\times10^{-317}      $& $0.05467$\\
\hline
\end{tabular}
\end{table}

\begin{table}[H]
\centering
\caption{List of numerical results vs. theoretical values, for $\beta=n+\frac{1}{2},n\in\mathbb{N}^{+}$}
\begin{tabular}{|cccllr|cccllr|}
\hline
$\alpha$        &$\beta$              &Size $N$      & Numerical $\lambda_{N}$ & Theoretical $\lambda_{N}$    & error     \\  \hline
$ -\frac{3}{4} $&$     \frac{3}{2}   $&  100         &$7.8618\times10^{-36}     $&$8.1371\times10^{-36}      $& $0.03502$\\
                &                     &  200         &$3.3759\times10^{-58}     $&$3.4638\times10^{-58}      $& $0.02605$\\
                &                     &  300         &$5.3913\times10^{-77}     $&$5.5084\times10^{-77}      $& $0.02172$\\
                &                     &  400         &$8.6803\times10^{-94}     $&$8.8455\times10^{-94}      $& $0.01902$\\
                &                     &  500         &$3.1289\times10^{-109}    $&$3.1825\times10^{-109}     $& $0.01713$\\
                &                     &  600         &$1.1225\times10^{-123}     $&$1.1402\times10^{-123}    $& $0.01571$\\
                &                     &  700         &$2.4313\times10^{-137}     $&$2.4668\times10^{-137}    $& $0.01459$\\
                &                     &  800         &$2.2724\times10^{-150}     $&$2.3035\times10^{-150}    $& $0.01368$\\
                &                     &  900         &$7.2147\times10^{-163}     $&$7.3079\times10^{-163}    $& $0.01292$\\
                &                     & 1000         &$6.5115\times10^{-175}     $&$6.5914\times10^{-175}    $& $0.01227$\\
$  -\frac{1}{2}$&$   \frac{5}{2}     $&  100         &$1.7527\times10^{-68}     $&$1.8456\times10^{-68}      $& $0.05303$\\
                &                     &  200         &$1.6233\times10^{-121}     $&$1.6918\times10^{-121}    $& $0.04223$\\
                &                     &  300         &$1.9900\times10^{-169}     $&$2.0635\times10^{-169}    $& $0.03693$\\
                &                     &  400         &$2.8604\times10^{-214}     $&$2.9564\times10^{-214}    $& $0.03356$\\
                &                     &  500         &$5.9391\times10^{-257}     $&$6.1240\times10^{-257}    $& $0.03114$\\
                &                     &  600         &$5.5030\times10^{-298}     $&$5.6642\times10^{-298}    $& $0.02928$\\
                &                     &  700         &$1.0780\times10^{-337}     $&$1.1080\times10^{-337}    $& $0.02779$\\
                &                     &  800         &$2.6691\times10^{-376}     $&$2.7400\times10^{-376}    $& $0.02655$\\
                &                     &  900         &$5.7436\times10^{-414}     $&$5.8901\times10^{-414}     $& $0.02550$\\
                &                     & 1000         &$8.0870\times10^{-451}     $&$8.2859\times10^{-451}     $& $0.02460$\\
$  2           $&$  \frac{5}{2}      $&  100         &$3.5614\times10^{-69}     $&$3.5767\times10^{-69}        $& $0.00428$\\
                &                     &  200         &$3.9629\times10^{-122}     $&$4.0207\times10^{-122}      $& $0.01460$\\
                &                     &  300         &$5.3740\times10^{-170}     $&$5.4660\times10^{-170}      $& $0.01711$\\
                &                     &  400         &$8.2704\times10^{-215}     $&$8.4184\times10^{-215}      $& $0.01790$\\
                &                     &  500         &$1.8070\times10^{-257}     $&$1.8397\times10^{-257}      $& $0.01809$\\
                &                     &  600         &$1.7433\times10^{-298}     $&$1.7748\times10^{-298}      $& $0.01803$\\
                &                     &  700         &$3.5306\times10^{-338}     $&$3.5937\times10^{-338}      $& $0.01787$\\
                &                     &  800         &$8.9908\times10^{-377}     $&$9.1495\times10^{-337}      $& $0.01766$\\
                &                     &  900         &$1.9823\times10^{-414}     $&$2.0168\times10^{-414}     $& $0.01743$\\
                &                     & 1000         &$2.8512\times10^{-451}      $&$2.9002\times10^{-451}     $& $0.01719$\\
$      0       $&$ \frac{7}{2}       $&  100         &$7.0602\times10^{-89}     $&$6.8217\times10^{-89}      $& $0.03379$\\
                &                     &  200         &$9.5989\times10^{-164}     $&$9.3315\times10^{-164}      $& $0.02786$\\
                &                     &  300         &$1.5672\times10^{-233}     $&$1.5286\times10^{-233}      $& $0.02464$\\
                &                     &  400         &$3.4926\times10^{-300}     $&$3.4140\times10^{-300}      $& $0.02251$\\
                &                     &  500         &$1.3836\times10^{-364}     $&$1.3546\times10^{-364}      $& $0.02096$\\
                &                     &  600         &$3.0405\times10^{-427}     $&$2.9804\times10^{-427}      $& $0.01976$\\
                &                     &  700         &$1.7488\times10^{-488}     $&$1.7160\times10^{-488}      $& $0.01879$\\
                &                     &  800         &$1.5603\times10^{-548}     $&$1.5322\times10^{-548}      $& $0.01798$\\
                &                     &  900         &$1.4700\times10^{-607}     $&$1.4446\times10^{-607}      $& $0.01729$\\
                &                     & 1000         &$1.0901\times10^{-665}      $&$1.0719\times10^{-665}      $& $0.01669$\\
\hline
\end{tabular}
\end{table}

\section{Appendix}

The integral identities listed below, which are relevant to our derivation and can be found in \cite{new1}, \cite{new2} and \cite{C21}.

\begin{equation}
\int_{a}^{b}\frac{dx}{\sqrt{(b-x)(x-a)}}=\pi.
\end{equation}
\begin{equation}
\int_{a}^{b}\frac{dx}{(x+t)\sqrt{(b-x)(x-a)}}=
\frac{\pi}{\sqrt{(t+a)(t+b)}}.
\end{equation}
\begin{equation}
\int_{a}^{b}\frac{\log (x+t)}{\sqrt{(b-x)(x-a)}}dx=
2\pi\log\left(\frac{\sqrt{t+a}+\sqrt{t+b}}{2}\right).
\end{equation}
\begin{equation}
\int_{a}^{b}\frac{\log (1-x)}{\sqrt{(b-x)(x-a)}(x+t)}dx=
\pi\frac{\log\left(\frac{(t+1)^{2}-\left(\sqrt{(t+a)(t+b)}-\sqrt{(1-a)(1-b)}\right)^{2}}{\left(\sqrt{t+a}+\sqrt{t+b}\right)^{2}}\right)}{\sqrt{(t+a)(t+b)}}.
\end{equation}

Gauss' recursion functions (\cite{new2}, P1019, 9.137$^{7}$, 1):
\begin{equation}\label{0.0}
\begin{split}
&\gamma\left[\gamma-1-(2\gamma-\alpha-\beta-1)z\right]F(\alpha,\beta;\gamma;z)+(\gamma-\alpha)(\gamma-\beta)zF(\alpha,\beta;\gamma+1;z)\\
&+\gamma(\gamma-1)(z-1)F(\alpha,\beta;\gamma-1;z)=0.
\end{split}
\end{equation}

\section{Acknowledgements}
The financial support of the Macau Science and Technology Development Fund under grant
number FDCT 130/2014/A3 and FDCT 023/2017/A1 are gratefully acknowledged. We would also like
to thank the National Science Foundation (NSF): CCF-1525754 and the University of Macau for generous support: MYRG 2014-00011 FST, MYRG
2014-00004 FST.

\end{document}